\documentclass{article}

\usepackage{graphicx}
\usepackage{latexsym,epsfig}
\usepackage{amsthm}
\usepackage{amssymb}
\usepackage{amsmath}
\usepackage{comment}
\usepackage{epstopdf}

\newcommand{\cH}{\ensuremath{\mathcal{H}}}

\def\BH{{\mathcal   B(\mathcal   H)}}
\def\dom{{\mathcal   D}}

\def\VH{{\mathcal   V(\mathcal   H)}}
\def\V0H{{{{\mathcal   V}}_{{\dom_{0}}}{({\mathcal   H})}}}

\def\cH{{{\mathcal   H}}}

\def\clos#1{\overline{#1}}

\def\Cal{\mathcal  }

\def\op{\overline{\oplus}}

\def\VH{{\mathcal   V(\mathcal   H)}}

\def\B{{\mathcal   B}_{f}({\mathcal   H})}
\def\V{{\mathcal   V}_{f}({\mathcal   H})}
\def\Vp{{\mathcal   V}_{f}({\mathcal   H})}
\def\S{{\mathcal   S}_{f}({\mathcal   H})}
\def\C{{\mathcal   C}_{f}({\mathcal   H})}
\def\G{{\mathcal   G}_{f}({\mathcal   H})}
\def\R{{\mathcal   R}_{f}({\mathcal   H})}
\def\VD{{\mathcal   V}_{fD}({\mathcal   H})}

\newcommand{\be}{\begin{equation}}
\newcommand{\bea}{\begin{eqnarray}}
\newcommand{\eea}{\end{eqnarray}}

\newcommand{\ed}{\end{document}}

\newtheorem{theorem}{Theorem}[section]

\theoremstyle{plain}

\newtheorem{corollary}[theorem]{Corollary}
\newtheorem{lemma}[theorem]{Lemma}
\newtheorem{proposition}[theorem]{Proposition}

\theoremstyle{definition}
\newtheorem{definition}[theorem]{Definition}
\newtheorem{example}[theorem]{Example}

\newtheorem{remark}[theorem]{Remark}

\begin{document}
\title%[On Bilinear Forms from the Point of View GEAs]
{On Bilinear Forms from the Point of View of Generalized Effect Algebras}
\author%[Anatolij Dvure\v{c}enskij and Ji\v{r}\'i  Janda]
{Anatolij Dvure\v{c}enskij$^{1,2}$ and Ji\v{r}\'i  Janda$^3$}
\maketitle
\begin{center}  \footnote{


The authors acknowledge the support (A.D.) by the Slovak Research and Development Agency under the contract APVV-0178-11, the grant VEGA No. 2/0059/12 SAV, ESF
Project CZ.1.07/2.3.00/20.0051, (J.J.) ESF
Project CZ.1.07/2.3.00/20.0051 and Masaryk University grant 0964/2009  }
Mathematical Institute,  Slovak Academy of Sciences,\\
\v Stef\'anikova 49, SK-814 73 Bratislava, Slovakia\\
$^2$ Depart. Algebra  Geom.,  Palack\'{y} Univer.\\
17. listopadu 12,
CZ-771 46 Olomouc, Czech Republic\\
$^3$ 	Depart. of Math. and Statistics, Faculty of Science\\
Masaryk University, Kotl\'a\v{r}sk\'a 267/2, CZ-611 37  Brno, Czech Republic
E-mail: {\tt
dvurecen@mat.savba.sk} \qquad {\tt xjanda@math.muni.cz}
\end{center}

\begin{abstract}
We study positive bilinear forms on a Hilbert space which are neither not necessarily bounded nor induced by some positive operator. We show when different families of bilinear forms can be described as a generalized effect algebra. In addition, we present families which are or are not monotone downwards (Dedekind upwards) $\sigma$-complete generalized effect algebras.
\end{abstract}

\section{Introduction}\label{aba:sec1}

The notion of an effect algebra was presented by Foulis and Bennett in \cite{FoBe}
as a model of quantum mechanical measurement. It allows an elegant algebraic description of the set $\mathcal {E(H)}$ of all Hermitian operators between the zero and
identity operator in a complex Hilbert space $\mathcal   {H}$ with a primary notion $\oplus$ which denotes the disjunction or the sum of two mutually excluding elements.  Effect algebras were inspired by D-posets, \cite{KoCh}, where the primary notion is a difference of two comparable events. Both structures are equivalent, and both are bounded posets.

Many examples of effect algebras are intervals in  partially ordered groups (= po-groups), for example, $\mathcal {E(H)}$ is the interval $[O,I]$ in the po-group $\mathcal{B(H)}_{sa}$ of all Hermitian operators on $\mathcal H.$    If the studied structure has no  greatest element, we can speak about generalized effect algebras. This situation can happen for an infinite-dimensional Hilbert space $\mathcal H$: Let $\mathcal{B(H)},$ $\mathcal {T(H)},$ and $\mathcal {C(H)}$ be the class of all bounded operators,  trace class operators, and compact operators, respectively,  on $\mathcal H$. They are partially ordered groups and for their positive cones $\mathcal{B(H)}^+,$  $\mathcal {T(H)}^+,$ and $\mathcal {C(H)}^+$ we have that they are generalized effect algebras which have no greatest element. The same is true for the class of positive trace operators under the identity or for the class of positive compact operators less than the identity. On the other hand, the class  of all positive bounded operators under the identity operator $I$ has the greatest element, namely $I.$

We can be recommend the monograph \cite{dvurec} as a basic source of information on effect algebras and generalized effect algebras.

Unbounded operators are of crucial importance for the Hilbert space quantum mechanics, see \cite{BEH, dvurec, rs,rs2,glea}. The situation with unbounded operators is from the algebraic point of view more complicated than one with bounded operators, see e.g. \cite{considerable, pr, romp1, romp3}, where the structure of the generalized effect algebra  $\VH$ of all positive linear operators and its properties was described.

The class of positive bilinear forms is more larger than the class of positive unbounded operators because we have bilinear forms which are not determined by any operator. Therefore, in this note we concentrate to the study of positive bilinear forms from the point of view of generalized effect algebras in a similar way as it was done for different sets of positive unbounded linear operators. We describe a structure of a generalized effect algebra on the set of positive bilinear forms as well as of important sub-classes of positive bilinear forms such as the class of regular and singular ones,  the class of closed ones, etc.

Moreover, some results on convergence  of sequences of bilinear forms which are monotone with respect to the partial order induced by the generalized effect algebraic partial operation in the given sub-class are given.

The paper is organized as follows. In Section 2, we recall some necessary definitions from the theory of effect algebras and Hilbert spaces, Section 3. Section 4 describes a structure of the generalized effect algebra on the set of all positive bilinear forms. We define some important subfamilies of positive bilinear forms, namely  of regular, singular, closed bilinear forms and ones generated by linear operators describing a generalized effect algebra structure on each of them. We also investigate a more restricted operation of sum of bilinear forms. In Section 5, we show which of the considered structures is or is not a monotone downwards (Dedekind upwards) $\sigma$-complete generalized effect algebra.

\section{Elements of generalized effect algebras}%2

In the section we gather  some basic definitions and facts on generalized effect algebras and effect algebras. For any unexplained notion on effect algebras or generalized effect algebras, we recommend the book \cite{dvurec}.

\begin{definition}\label{def:GEA}
A partial algebra $(E;\oplus,0)$ is called a {\em generalized
effect algebra} if $0\in E$ is a distinguished element and
$\oplus$ is a partially defined binary operation on $E$ which
satisfy the following conditions for any $x,y,z\in E$:

\begin{description}%[\hbox{\rm(GEiii)}]
\setlength{\leftmargin}{1.8cm}%\width{\phantom{\rm(GEiii)}}}
\item[\rm(GEi)\phantom{ii}]  $x\oplus y=y\oplus x$, if one side is defined,
\item[\rm(GEii)\phantom{i}]
$(x\oplus y)\oplus z=x\oplus (y\oplus z)$, if one side is defined,
\item[\rm(GEiii)] $x\oplus 0=x$,
\item[\rm(GEiv)\phantom{i}] $x\oplus y=x\oplus z$ implies $y=z$
(cancellation law),
\item[\rm(GEv)\phantom{ii}]
$x\oplus y=0$ implies $x=y=0$.
\end{description}
\end{definition}

In every generalized effect algebra $E,$ a partial binary relation $\leq:=\leq_\oplus$ can be defined by
\begin{description}
\item[\rm(ED)] $x\leq y$ iff there exists an element $z\in E$ such that $x \oplus z$ is defined and $x \oplus z=y$.
\end{description}
Then $\leq$ is a partial order on $E$ under which $0$ is the least
element of $E$.

By (GEiv), the element $z$ such that $x\oplus z = y$ is unique and we denote it by $z = y \ominus x.$

A generalized effect algebra with the top element $1 \in E$ is called an {\it effect algebra} and we usually write $(E;\oplus,0,1).$ For example, if $\mathcal {E(H)}$ denotes the class of all Hermitian operators between the zero operator, $O$, and the identity operator, $I,$ then $(\mathcal{E(H)};\oplus, O,I),$ where $A\oplus B $ is defined  whenever $A+B \le I,$ and then $A\oplus B:= A+B,$ is a prototypical example of an effect algebra.

More generally, let $G$ be an Abelian po-group, i.e. a group $G$ is endowed with a partial order $\le$ such that if $a\le b,$ then $a + c\le b + c$ for every $c \in G.$ The positive cone $G^+$ of a po-group $G$ is defined as the set $G^+=\{g \in G \mid g\ge 0\}.$ Endow $G^+$ with the total operation $+,$ then $(G^+;+,0)$ is an example of a generalized effect algebra. If $u>0$,  define  an interval $[0,u]=\{g \in G \mid 0\le g \le u\}$ endowed with the partial operation $\oplus$ defined by $a\oplus b = a + b$ whenever $a+b\le u,$ then $([0,u];\oplus, 0,u)$ is an effect algebra. Also the set $[0,u) = \{g \in G \mid 0\le g < u\},$ with
the partial operation $\oplus$ by $a\oplus b = a + b$ whenever $a+b < u,$ forms a generalized effect algebra $([0,u];\oplus, 0)$.

A subset $Q$ of $E$ is called a {\it sub-generalized effect algebra}
of $E$ iff (i) $0 \in Q$, (ii)
if out of elements $x,y,z \in E$ such that $x\oplus y=z$ at least two are in $Q$, then all $x,y,z \in Q$.

Let $(E;\oplus,0)$ be a generalized effect algebra and $Q \subseteq E$ its subset. By  $\oplus_{\mid Q}$ we will denote the
restriction of $\oplus$ onto $Q\times Q $. That is, $\oplus_{\mid Q}$ is a partial operation on $Q$ such that, for any $x,y\in Q, x \oplus_{\mid Q} y$ is defined if and only if
$x \oplus y$ is defined in $E$ and $x\oplus y \in Q$ and then we set $x \oplus_{\mid Q} y := x \oplus y$.

\begin{remark}\label{subrem}
Let $(E;\oplus,0)$ be a generalized effect algebra and $S \subseteq E$ be a subset of $E$ such that $0\in S$ and whenever $x \oplus y$ is defined in $E$ for some $x, y \in S,$ then
$x \oplus y \in S$. Then $(S;\oplus_{\mid S},0)$ is a generalized effect algebra on its own with the induced partial order $\le := \le_{\mid S}$, but it need not to be a sub-generalized effect algebra of $(E;\oplus,0),$ because the definition of a sub-generalized effect algebra means that if $a,b \in S,$ then $a\le_\oplus b$ iff $a\le_{\mid S} b.$ In general, $a\le_{\mid S} b$ implies $a\le_\oplus b,$ but the converse statement is not necessarily true. For example, let $E=\mathbb Z^+$ and $S =\{0,4,6,8,10,\ldots\}.$ Then $E$ and $S$ are generalized effect algebras with respect to the standard addition, but $S$ is not a sub-generalized effect algebra of $E$ while $4 \le_\oplus 6$ but $4 \not\le_{\mid S} 6.$

We underline that every sub-generalized effect algebra $Q \subseteq E$ is a generalized effect algebra $(Q;\oplus_{\mid Q},0)$ in its own right.
\end{remark}

\section{Basic definitions and results on a Hilbert space and operators}%3

A complex vector space $S$ with an inner product $(\cdot\,,\cdot)$  is said to be a {\it pre-Hilbert space}. If
$\mathcal H$ is a complex inner product space which is complete with respect to the induced metric $\left\|\cdot\right\|,$ we call it a {\it Hilbert space}. We will use a "mathematical notion" of the inner product, hence $(\cdot\,,\cdot)$ is
linear in the left argument and antilinear in the right one.

We will assume that every  linear operator $A$ on $\mathcal H$ (i.e., a linear map $A:D(A)\to\Cal H$) has
the domain $D(A)$ which is a linear subspace dense in $\Cal H$ with respect to the topology induced by the inner product (we are saying that
$A$ is {\em densely defined}). As it was already said, by the symbol $O$ we mean the null operator; it is a bounded operator.
By $\Cal {O(H)}$ and $\Cal {D(H)}$ we denote the set of all operators on $\mathcal H$ and the set of dense linear subspaces of $\mathcal   H,$ respectively. A symbol $\clos{S}$ denotes the closure of a given subset $S$ of $\mathcal H.$ An operator $A$ is said to be {\it positive} if
$(Ax,x) \ge 0$ for all $x\in D(A)$.

For more information on the Hilbert space theory and Hilbert space operators we recommend to consult the books \cite{BEH, rs, rs2,Hal, kato}.

For every linear operator $A:D(A)\to\Cal H$ with
$\clos{D(A)}=\Cal H,$ there exists an {\em adjoint
operator $A^*$ of $A$} such that
$D(A^*)=\{y\in\Cal H\mid$ there exists a vector
$y^*\in\Cal H$ such that $(x,y^*) =(Ax,y) $\ for every $x\in D(A) \}$
and $A^*y=y^{*}$ for every $y\in D(A^{*})$.
If $A^*x=Ax$ for all $x\in D(A)$, $A$ is called {\it symmetric} and if also $D(A^*)=D(A)$, we call $A$ {\em self-adjoint} or {\it Hermitian} (for more details see \cite{BEH}).

Recall that a linear operator $A:D(A)\to\cH$ is said to be (i)
{\it bounded } if there exists a real constant $c\ge0$ such that $\|Ax\|\le c\|x\|$ for all $x\in
D(A),$ and (ii) {\it unbounded} if, for every $c\in \mathbb R$, $c\ge0,$ there exists a vector $x_c\in D(A)$ with
$\|Ax_c\|>c\|x_c\|$.  The set of all bounded operators on $\mathcal   H$ is denoted by ${\BH}$. Then $\BH$ with the usual addition $+$ of bounded operators and with the usual ordering of operators $A\le B$ iff $(Ax,x)\le (Bx,x)$ for any $x \in \mathcal H$ is a po-group, as well as the class of all Hermitian operators, $\BH_{sa}.$ In addition, $\mathcal{E(H)}=\{A \in \BH_{sa} \mid O \le A\le I\},$ and the positive cones $\BH^+$ and $\BH_{sa}^+$ induce generalized effect algebras.

We say that an operator $B$ with the domain $D(B)$ is an {\it extension} of an operator $A$ with the domain $D(A)$ if $D(A) \subseteq D(B)$ and $Ax=Bx$ for any $x \in D(A).$

\begin{theorem}\label{ext}{\rm \cite[Thm 1.5.5.]{BEH}}
For every bounded operator $A:D(A)\to\cH$ densely defined on $D(A)\subset \mathcal   H,$ there exists
a unique bounded extension $B$ such that
$D(B)=\cH$ and $Ax=Bx$ for every $x\in D(A)$. %Such $B$ is also a bounded operator.
\end{theorem}

We  define the set
$$
\VH = \{A:D(A)\rightarrow{\Cal H}\mid A \geq O, \overline{D(A)}={\mathcal   H} \text{ and}\ D(A)=\cH\ \text{if}\ A\ \text{is bounded}\}.
$$
Let us define a partial operation $\oplus_{\dom}$ on $\VH$ as follows: for every $A, B \in \VH$, $A \oplus_{\dom} B$ is defined and $A \oplus_{\dom} B = A+B$ iff $A$ or $B$ is bounded or $D(A)=D(B)$.
In \cite{romp1}, it has been shown that $(\VH; \oplus_{\dom}, O)$ is a generalized effect algebra.

Let $D(t) \in\mathcal   {D(H)}$, $\clos{D(t)} = {\mathcal   H}$, a map $t : D(t) \times D(t) \rightarrow \mathbb{C}$ is called a {\it bilinear form} (or a {\it sesquilinear form}) if and only if it is
linear in both arguments and $(\alpha x, \beta y) = \alpha{\overline \beta} (x,y)$ for all $\alpha, \beta \in \mathbb{C}$, $x,y \in D(t)$, where $\overline{\beta}$ means the complex conjugation of $\beta.$  A complex-valued function $\hat t$ with the definition domain $D(t)$ is defined by $\hat t(x)=t(x,x),$ $x \in D(t)$ and it is said to be a {\it quadratic form} induced by a bilinear form $t,$ \cite[p. 12]{Hal}.
A bilinear form $t$ is called {\it symmetric} if $t(x,y) = \overline{t(y,x)}$. The set $R(t) = \{t(x,x) \mid x \in D(t), \left\| x \right\|=1\} \subseteq \mathbb{C}$ is called the  {\it numerical range} and it is well known that
$t$ is symmetric if and only if $R(t) \subseteq \mathbb{R}$. A symmetric bilinear form $t$ is (i) {\it semibounded} if there exists $m_{t} \in \mathbb{R}$ such that $m_t = \inf R(t),$ and  (ii) {\it positive}
if $m_{t} \ge 0$. Let $\mathcal {PBF(H)}$ be the class of all positive bilinear forms. A semibounded bilinear form is called {\it bounded} if there is $n_t \in \mathbb{R}$ such that $n_t = \sup R(t)$. Otherwise $t$ is called {\it unbounded}.
There exists a bilinear form $o$ with $D(o)=\mathcal H$ defined by $o(x,y)=0$ for all $x,y \in \mathcal   H.$ We denote by $\mathcal{BF(H)}$ the set of all bilinear forms on ${\mathcal H}.$

Given a symmetric semibounded bilinear form $t,$ we can equip its domain $D(t)$ with
an inner product $(x, y)_t := t(x, y) + (1+m_t)(x, y)$. In this way $D(t)$ becomes a pre-Hilbert space. Whenever $D(t)$ with $(x, y)_t$ is a Hilbert space, we call $t$ {\it closed.}

A bilinear form $s$ is an {\it extension} of a bilinear $t$ iff $D(t) \subseteq D(s)$ and, for every $x,y \in D(t)$, $t(x,y) = s(x,y);$ we will write $s_{\mid D(t)}=t.$ A bilinear form $t$ is {\it closable} if it has some closed extension.

Given two densely defined positive bilinear forms $t$ and $s,$ we write
\begin{eqnarray} t \preceq s \label{prec}
\end{eqnarray}
if and only if $D(t)\supseteq D(s)$ and $t(x,x)\le s(x,x)$ for all $x \in D(s).$ Then
$\preceq$ is a partial order on $\mathcal{PBF(H)},$ and the bilinear form $o$ is the least element of the class $\mathcal {PBF(H)},$ i.e. $o \preceq t$ for any $t \in \mathcal{PBF(H)}.$

It is well known that there is a one-to-one correspondence between  bounded linear operators and bounded bilinear forms. That is, for any bounded bilinear form $t$, $D(t)={\mathcal   H},$ there exists unique $A \in \BH$ given by $t(x,y) = (Ax,y)$ for all $x,y\in {\mathcal   H}$. Conversely, for any $B \in \BH,$ the mapping $s:\mathcal H \times \mathcal H \to \mathbb C,$ defined by $D(s) =\mathcal H$ and $s(x,y) = (Bx,y)$ for all $x,y\in {\mathcal   H}$, is a bilinear form.

It can be also seen that, for any $A \in \VH$, the map $(Ax, y)$ is a bilinear form on $D(A)$. We say that a bilinear form $t$
{\it corresponds} to an operator $A\in  \VH$  iff $D(A) \subseteq D(t)$ and, for every $x, y \in D(A),$ $t(x,y)=(Ax,y)$. We say that $t$ is {\it generated} by $A$ if $t$ corresponds to $A$ and $D(A)=D(t)$.

\section{Generalized effect algebras of positive bilinear forms}%4

The present section is the heart of the paper. We study different generalized effect algebras consisting of  families of positive bilinear forms.

The following is a well known fact, see e.g. \cite[Proposition 3.1.1.]{BEH}.

\begin{proposition}\label{extbeh}
There is a one-to-one correspondence between bounded linear operators and bounded bilinear forms on $\mathcal   H$ given by $t(x,y) = (Ax,y)$ for some $A \in \BH$ and all $x,y \in {\mathcal   H}$.
\end{proposition}

\begin{lemma}\label{funct}{\rm \cite[Corollary 1.3.7]{glea}}
For any bounded linear functional $f$ on some inner product space $S$ (not necessary complete), there exists a unique vector $z\in {\overline S}$ such that $f(y) = (y,z)$ for all $y \in S$.
\end{lemma}

\begin{lemma}\label{exbil}
Let $t$ be a positive bounded bilinear form defined on some dense subspace $D(t)$ of ${\mathcal   H}$. Then there exists a unique bounded bilinear form
$s$ such that $D(s) = {\mathcal   H}$ and $s_{\mid {D(t)}}=t$.
\end{lemma}

\begin{proof}
Let $t$ be a positive bounded bilinear form defined on some dense subspace $D(t) \subseteq {\mathcal   H}$. Every $x \in D(t)$
defines a linear functional $t_x$ by $t_{x}(y) = t(y,x)$. By Lemma \ref{funct}, for every $t_{x}$
there exists $z \in {\mathcal   H}$ such that $t_{x}(y) = (y,z) = t(y,x)$ for all $y \in D(t)$. We can define a map $A : D(t) \rightarrow {\mathcal   H}$
by $Ax=z$. We show that $A$ is an operator.

Let us have  arbitrary $x_1,x_2 \in D(t)$. There are $z_1, z_2, z \in {\mathcal   H}$ such that

\begin{eqnarray*}
&  & Ax_1=z_1 \text{ hence } t(y,x_1)=(y,z_1) \text{ for all } y\in D(t),\label{eq} \\
&  & Ax_2=z_2 \text{ hence } t(y,x_2)=(y,z_2) \text{ for all } y\in D(t), \\
&  & A(x_1+x_2) = z \text{ hence } t(y,x_1+x_2)=(y,z) \text{ for all } y\in D(t).
\end{eqnarray*}

Then $(y,z) = t(y,x_1+x_2)=t(y,x_1)+t(y,x_2)=(y,z_1)+(y,z_2)=(y,z_1+z_2)$ for all $y\in D(t)$, hence $z = z_1+z_2$
and $A(x_1+x_2) = Ax_1+ Ax_2$.

Similarly, choose $x \in D(t)$ and arbitrary $\alpha \in \mathbb{C}$. Let $z_1, z_2 \in {\mathcal   H}$ be  elements such that

\begin{eqnarray*}
& & Ax = z_1 \text{ hence } t(y,x) = (y,z_1) \text{ for all } y\in D(t), \\
& & A(\alpha x) = z_2 \text{ hence } t(y,\alpha x) = (y,z_2) \text{ for all } y\in D(t).
\end{eqnarray*}
We have $(y,z_2) = t(y,\alpha x)= {\overline \alpha}t(y,x) = {\overline \alpha} t(y,z_1) = t(y,\alpha z_1)$ for all $y\in D(t)$, which yields $z_2 =\alpha z_1,$
that is $\alpha Ax = A(\alpha x)$.

Hence, $t(y,x) = (y, Ax)$ for all $x,y \in D(t)$ and, for any $x,y \in D(t),$ we have $(y, Ax) = t(y,x) = \overline{t(x,y)}=\overline{(x,Ay)}=(Ay,x).$ We have $t(x,y)=(Ax,y),$ $x,y \in D(t).$ Therefore, $A$ is an operator and
$t$ is generated by  $A$.

Since $A$ is a bounded operator, it can be uniquely extended to an operator ${\overline A}$ defined on ${\mathcal   H}$, i.e. $D(\overline{A})={\mathcal   H}$. The operator ${\overline A}$ generates a bilinear form $s$, which is an extension of $t$ defined on the whole ${\mathcal   H}$. If there exists another extension $q$ of $t$ on ${\mathcal   H}$, by Proposition \ref{extbeh}, there exists an operator $B$ which generates $q$. We have ${\overline A}_{\mid{D(t)}} = A = B_{\mid D(t)},$ hence by Theorem \ref{ext}, ${\overline A} = B$.
\end{proof}

On the set of all  bilinear forms, $\mathcal{BF(H)},$ we can define a {\it usual sum} $t+s$ for any $t, s$ on $ D(t+s):=D(t) \cap D(s)$
by $(t+ s)(x,y):=t(x,y)+ s(x,y)$ for all  $x,y \in D(t) \cap D(s)$, and the multiplication by a scalar $\alpha\in \mathbb C$ by $(\alpha t)(x,y):= \alpha t(x,y)$ for $x,y \in D(\alpha t):=D(t).$

\begin{lemma}\label{nr}
Let $t$ and $s$ be  positive bilinear forms. Then, for the numerical range $R(t+s)$, we have $R(t+s) \subseteq \{x + y \mid x \in R(t),\ y \in R(s)\}$.
\end{lemma}

\begin{proof}
Clear.
\end{proof}

Using  Lemma \ref{exbil}, the definition of a class of bilinear forms $\Vp$ in the next theorem is an analogy with the definition of the set of operators $\VH$ from \cite{romp1}.

\begin{theorem}

Let ${\mathcal   H}$ be an infinite-dimensional complex Hilbert space. Let us define the set of bilinear forms
\[
\Vp = \{t \mid t \in \mathcal{BF(H)}, \overline{D(t)}={\mathcal   H} \text{ and} \text{ if }\ t\ \text{is bounded, then } D(t)=\cH \}.
\]

\noindent Let us define a partial operation $\oplus$ on $\Vp$ by
\begin{description}
\item for  $t,s \in \Vp$, $t \oplus s$ is defined if and only if $t$ or $s$ is bounded or $D(t) = D(s)$ and then
$t \oplus s = t + s$.
\end{description}
Then $(\Vp;\oplus,o)$ is a generalized effect algebra.
\end{theorem}

\begin{proof}
Let us have two bounded bilinear forms $t,s \in \Vp$. Then $t\oplus s = t+s$ is by Lemma \ref{nr} also bounded and $D(t+s)= D(t) \cap D(s) = {\mathcal   H}$, hence $t\oplus s \in \Vp$. If
$t$ is unbounded and $s$ bounded, we have $t\oplus s = t+s$ is unbounded with $D(t+s)=D(t) \cap {\mathcal   H} = D(t)$, i.e. $t\oplus s$ is densely defined and $t\oplus s \in \Vp$.
Whenever $t$ and $s$ are both unbounded, then $t\oplus s$ is defined iff $D(t) = D(s)=D(t+s)=D(t\oplus s),$ hence also $t\oplus s \in \Vp$.

The commutativity of $\oplus$ holds since the usual sum $+$ is commutative. (GEiii) is also clear because $o$ is bounded and $D(o) = {\mathcal   H}$.

Let us show the positivity (GEv) of the operation $\oplus$ on $\Vp.$ The proof follows ideas from \cite{pr}.
Let us have $s,t\in \Vp$ such that $s\oplus t=o$. Then $D(t\oplus s)=D(o)=\mathcal H$ and therefore, $D(t)=D(s)=\mathcal H.$
We have $t(x,x)+s(x,x) = 0$ for all $x \in D(s) \cap D(t)=\mathcal H.$ So that, $t(x,x)=0=s(x,x)$ for any $x \in \mathcal H.$  For any $x, y \in \mathcal H,$ it holds $0=t(x+y,x+y)=t(x,x)+t(x,y)+t(y,x)+t(y,y) = t(x,y)+t(y,x)$ and similarly $0=t(x+\imath y,x+\imath y)= t(x,x)+\imath t(x,y)+\imath t(y,x)+t(y,y)=\imath (t(x,y)-t(y,x))$, hence $t(x,y)=t(y,x)=0$ and $t=o.$

(GEii) Let us have $r,s,t \in \mathcal{BF(H)}$ such that $(r\oplus s)\oplus t$ is defined.

(i) Let $r,s,t$ be all unbounded, hence $D(r)=D(s)=D(t)$.
Then $(r\oplus s)\oplus t = (r(x,y)+s(x,y))+t(x,y) = r(x,y)+(s(x,y)+t(x,y)) = (r\oplus s)\oplus t$ for all $x,y \in D(t)$. The other cases are similar.

(ii) Let $r$ be bounded and $s, t$ unbounded. Then $D(r) \cap D(s) = D(s)$ and $(r\oplus s)\oplus t$ is defined iff $D(s)=D(t)$, hence $r \oplus (s \oplus t)$ is defined
and $(r\oplus s)\oplus t = r \oplus (s \oplus t)$ (the case when $s$ is bounded and $r, t$ are unbounded is similar).

(iii) Let $t$ be bounded, $r, s$ unbounded. Then $D(r) = D(s) = D(s \oplus t)$ and we have defined $r \oplus (s \oplus t)=(r\oplus s)\oplus t$.

(iv) It is not hard to see that if  two of three bilinear forms $r,s,t \in \Vp$ are bounded and the third one is unbounded, then the both sides of the equation of the associativity are defined and we
have the same domains which equal to the domain of the unbounded bilinear  form. The case when $t, r, s$ are all bounded immediately follows from the associativity of the usual sum (and of the associativity of complex numbers, respectively).

(GEiv) Let $t \oplus r = t \oplus s$.
Let $t, r, s$ be all bounded. Then $t(x,y) + r(x,y) = t(x,y)+s(x,y)$ for all $x,y \in {\mathcal   H}$, which is equivalent to $r(x,y)=s(x,y)$ for all $x,y \in {\mathcal   H}$
hence $r = s$.

Let $t$ be bounded and $r$  unbounded. Then $s$ has to be unbounded and $D(r) = D(t+r) = D(t+s)=D(s)$. By the same argument as in the previous case (restricted on $D(r)$) we have
$r=s$.

Let $t$ be unbounded and $r$ bounded. Then  $D(t) =D(t+r) = D(t+s)$. For every $x, y \in D(t),$ it holds $t(x,y) + r(x,y) = t(x,y)+s(x,y).$ Hence, $r(x,y)=s(x,y)$, that is
$s_{\mid D(t)} = r_{\mid D(t)}$. Since $r$ is on $D(t)$ bounded, $s$ is also bounded and by Lemma \ref{exbil}, $s$ can be extended on ${\mathcal   H}$ in a unique way, that is $s = r$.

Let $t$ be unbounded and $r$ unbounded. $D(t) = D(r) = D(t+r) = D(t+s)$ and in the same way as in the previous case, $s_{\mid D(t)} = r_{\mid D(t)}$. Since $r$ is unbounded, $s$ is also unbounded. Because $t \oplus s$ is defined, we have $D(t) = D(s) = D(r),$ hence $s = r$.
\end{proof}

\begin{theorem}\label{cc}{\rm \cite[Thm X.23]{rs2}}
Let $t$ be a bilinear form. If it is generated by some linear operator $A$, then $t$ is closable.
\end{theorem}

\begin{remark}\label{remord}
Let $\leq:=\leq_\oplus$ be the partial order derived from the generalized effect algebra $(\Vp;\oplus,o).$ Then $t\leq s$ means that there is an $r \in \Vp$ such that $t\oplus r = s.$ In addition, we have then also $t \preceq s,$ where $\preceq$ is defined by (\ref{prec}). We note that  the statement ``$t \le_\oplus s$" and the statement ``$t \preceq s$" and ``[$D(t)=\mathcal H$ or $D(t)=D(s)$]" are equivalent.
\end{remark}

\begin{theorem}\label{bs}{\rm \cite[Thm 2.1, Thm 2.2]{bs}}
Let $t$ be a densely defined positive symmetric bilinear form on a Hilbert space ${\mathcal   H}$. Then there exist two positive symmetric bilinear forms
$t_r$ and $t_s$ such that $D(t)=D(t_r)=D(t_s)$ such that
\begin{equation}
t = t_r + t_s,
\label{eq:rs}
\end{equation}
where $t_r$ is the largest closable bilinear form less than $t$ in the ordering $\preceq.$
\end{theorem}

\begin{remark}\label{sireg}
According to  Theorem \ref{bs}, for every positive symmetric bilinear form $t,$  the components $t_r$ and $t_s$ are said to be the {\it regular} and {\it singular} part of $t,$ respectively.

If $t_{s} = o,$  the bilinear form  $t$ is said to be {\it regular},    and if $t_{r}=o,$ $t$ is  said to be {\it singular}. Hence, the bilinear form $o$ is a unique positive symmetric bilinear form which is simultaneously regular and singular.
\end{remark}

In addition, let $t$ be a positive symmetric bilinear form
and  $b$ be  an everywhere defined bilinear form determined by a positive
Hermitian operator $B$ on $\mathcal H$, i.e., $b(x,y) = (Bx,y)$, $x,y \in \mathcal H$. Define $t + b$
on $D(t)$. Then \cite[Cor 2.3]{bs}
\begin{equation}
(t+b)_r = t_r + b;\quad (t + b)_s = t_s.
\label{eq:bound}
\end{equation}
Moreover, if $s\preceq  t$, then
\begin{equation}\label{eq:reg}
s_r \preceq t_r.
\end{equation}

We notice that for singular parts the statement if ``$s\preceq t$, then $s_s \preceq t_s$" does not hold, in general.

The following example presents a nonzero everywhere  defined positive singular symmetric bilinear form, see also \cite[Ex 1.6.14]{glea}.

\begin{example}
Let $\{e_n\}$ be an orthonormal basis of an infinite-dimensional
Hilbert space $\mathcal H$ which is a part of a Hamel basis $\{e_i\}_{i \in I}$ of $\mathcal H$
consisting of unit vectors from $\mathcal H$. Fix an element $e_{i_0}$, $i_0 \in I,$
which does not belong to the orthonormal basis $\{e_n\}$, and define a linear operator $T$ on $\mathcal H$ by

\begin{equation}
T(\sum_i \alpha_i \,e_i) = \alpha_{i_0} \, e_{i_0},
\end{equation}
where $\alpha_{i_0}$ is the scalar corresponding to $e_{i_o}$ in the decomposition $x = \sum_i \alpha_i \,e_i$, $x \in \mathcal H$, with  respect to the given Hamel basis. Then $T$ is an everywhere defined unbounded linear operator.

If we define a bilinear form $t$ with $D(t) = \mathcal H$ via
\begin{equation}
t(x,y) = (Tx,Ty),\ x,y \in \mathcal H,
\end{equation}
then $t$ is a nonzero everywhere  defined positive singular symmetric bilinear form.
\end{example}

A criterion  of singular bilinear forms, \cite{Lug}, see also \cite[Lem 1.6.13]{glea}, says: A positive symmetric bilinear form $t$ with a dense  domain $D(t)$ in  $\mathcal H$ is singular
if and only if, for any non-zero $x \in \mathcal H,$ there exists a vector $y \in D(t)$ such that
\begin{equation}
t(y,y) < |(x,y)|^2.
\end{equation}

Previous Theorem \ref{bs} is crucial. From Theorem \ref{cc} we know that every positive bilinear form which can be created by some densely defined positive linear operator
is closable. Hence, there do not exist positive linear operators that can define  singular bilinear forms different of $o,$ and the class of all positive bilinear forms is much richer than the class of positive operators.

\begin{definition}
Let ${\mathcal H}$ be an infinite-dimensional complex Hilbert space. Let us define following subsets of $\Vp$:
\begin{description}

\item $\B = \{t \mid t \in \Vp, t \text{ is bounded}\}$,
\item $\R = \{t \mid t \in \Vp, t \text{ is regular}\}$,
\item $\S = \{t \mid t \in \Vp, t \text{ is singular}\}$,
\item $\G = \{t \mid t \in \Vp, \exists \, A \in \VH \text{ such that } D(A) = D(t), t(x,y) = (Ax,y) \text{ for all } x,y \in D(t)\}$,
\item $\C = \{t \mid t \in \Vp, t \text{ is closed}\}$.

\end{description}
\end{definition}

\begin{theorem}\label{subbound}
Let ${\mathcal H}$ be an infinite-dimensional complex Hilbert space.  Then $(\B; \oplus_{\mid \B},o)$
is a sub-generalized effect algebra of the generalized effect algebra $(\Vp; \oplus,o).$
\end{theorem}

\begin{proof}
Let $r, s \in \B$, then we have defined $r\oplus s = r+s =t$. Since both $r$ and $s$ are bounded, with Lemma \ref{nr} so is $t$.
Whenever $r+s=t$ for some $t,r \in \B, s\in \Vp$, then
$s(x,x) = t(x,x) - r(x,x)$ for all $x \in \mathcal{H}$. Since $r$ is positive, we have $s(x,x)\leq t(x,x)$ for all $x \in \mathcal{H}$.
This gives $n_s \leq n_t$ for suprema $n_s, n_t$ of the numerical ranges of $s$ and $t$, hence $s$ is bounded.
\end{proof}

\begin{theorem}\label{regg}
Let ${\mathcal H}$ be an infinite-dimensional complex Hilbert space.  Then $(\R; \oplus_{\mid \R},o)$ is a generalized effect
algebra, but it is not a sub-generalized effect algebra of the generalized effect algebra $(\Vp; \oplus,o).$
\end{theorem}

\begin{proof} The sum of two closable bilinear forms is always closable \cite[Thm VI.1.31]{kato}. Hence, by Remark \ref{subrem}, $(\R; \oplus_{\mid \R},o)$ is a generalized effect
algebra. According to \cite[Cor 2.3., Rem (2)]{bs}, there exist nonzero bilinear forms $t,r \in \R$, $s \in \S$, $D(t) = D(s) = D(r)$ such that
$t \oplus s = r$, but we have $\R \cap \S = \{ o \}$.
\end{proof}

\begin{theorem}
Let ${\mathcal   H}$ be an infinite-dimensional complex Hilbert space.  Then $(\G;\oplus_{\mid \G},o)$ is a sub-generalized effect algebra of $(\Vp; \oplus, o)$ and is isomorphic to the generalized effect algebra $(\VH; \oplus_{\dom},0).$
\end{theorem}

\begin{proof} Clearly $o(x,y) = (Ox,y)$. Let us have $t, r \in \G$, $s \in \Vp$ such that $t \oplus r = s$. Then $D(t) \cap D(r) = D(s)$ and there exist $A, B \in \VH$ such that
$D(A) = D(t), D(B) = D(r)$ and $s(x,y) = t(x,y)+r(x,y)=(Ax,y)+(Bx,y)=((A+B)x,y)$ for all $x,y \in D(s)$, hence $s \in \G$. In the same way we deal with the case when $t,s \in \G, r \in \Vp$.
An isomorphism can be seen immediately from definitions.
\end{proof}

\begin{theorem}\label{closed}
Let ${\mathcal H}$ be an infinite-dimensional complex Hilbert space.  Then $(\C; \oplus_{\mid \C},o)$
is a generalized effect algebra.
\end{theorem}

\begin{proof}
We have $o$ is bounded and closed. According to \cite[Thm VI.1.31]{kato}, whenever $t_1, t_2 \in \C$ such that $t_1 \oplus t_2$ is defined,
also $t_1 \oplus t_2 \in \C,$ hence by Remark \ref{subrem}, $(\C; \oplus_{\mid \C},o)$
is a generalized effect algebra.
\end{proof}

Let us denote by $\mathcal{SA(H)} $ the set of all positive self-adjoint operators on $\mathcal H$ (clearly $\mathcal{SA(H)} \subseteq \VH$).

\begin{theorem}\label{one}{\rm \cite[Thm VIII.15]{rs}}
Let $t \in \Vp$ be a closed positive bilinear form. Then there exists a unique positive self-adjoint operator $A_t \in \mathcal{SA(H)}$ such that $t$ corresponds to $A_t$.
\end{theorem}

A closed bilinear form $t$  is said to be
{\it associated} with a unique self-adjoint operator $A_t$ from Theorem \ref{one}.

The previous Theorem \ref{one} is an important result from the theory of Hilbert spaces. In other words, it says that
there is a one-to-one correspondence between positive self-adjoint operators $\mathcal{SA(H)}$ and closed positive bilinear forms $\C$
(for more details, see \cite{rs}, \cite[Thm 4.6.8.]{BEH}).

Moreover, for closed positive bilinear forms $s_1, s_2$, if the usual sum $s_1+s_2 =: s$ is densely defined, i.e. $\overline{D(s)}=\mathcal{ H}$,
then $s$ is closed. According to \cite{BEH}, we can define a {\it form sum} ``$\dot{+}$" for associated self-adjoint operators $A_{s_1}, A_{s_2} \in \mathcal{SA(H)}$
by $A_{s_1} \dot{+} A_{s_2} = A_{s}$.

This allows us to extend our result about closed positive bilinear forms $\C$ to the set of all positive
self-adjoint operators $\mathcal{SA(H)}$ as follows.

\begin{proposition}
Let ${\mathcal H}$ be an infinite-dimensional complex Hilbert space.
 Let us define a partial operation $\dot{\oplus}$ on $\mathcal{SA(H)}$ by

\begin{description}
\item for $A_{s_1},A_{s_2} \in \mathcal{SA(H)}$, $A_{s_1} \dot{\oplus} A_{s_2}$ is defined if and only if $s_1 \oplus s_2$ is defined and then
$A_{s_1} \dot{\oplus} A_{s_2} = A_{s_1 + s_2}$,
\end{description}
\noindent where $s_1, s_1, s_1 + s_2$ are positive closed bilinear forms associated to positive self-adjoint linear operators $A_{s_1}$, $A_{s_2}$ and $A_{s_1 + s_2}$.

Then $(\mathcal{SA(H)};\dot{\oplus},O)$ is a generalized effect algebra.
\end{proposition}

\begin{proof}
The statement follows from Theorem \ref{closed} and Theorem \ref{one}.
\end{proof}

\begin{lemma}\label{restlem}
Let ${\mathcal   H}$ be an infinite-dimensional complex Hilbert space. For any $t, s \in \Vp$ such that $t \oplus s$ is defined, the following statement hold:

\begin{enumerate}
\item[\rm(i)\phantom{ii}] If $t$ is bounded, then $t = t_{r}.$
\item[\rm(ii)\phantom{i}] $(t + s)_{r}=t_{r} + s_{r}$ if and only if $(t+s)_{s} = t_{s}+s_{s}.$
\item[\rm(iii)] $t_{r} + s_{r} \leq (t+s)_{r}.$
\item[\rm(iv)] $t = s$ if and only if $t_r = s_r$ and $t_s = s_s.$
\end{enumerate}
\end{lemma}

\begin{proof}
(i) follows from the definition. For (ii), let $t,s \in \Vp$, $t \oplus s$ be defined and $(t + s)_{r}=t_{r} + s_{r}$. Then $(t_{r} + t_{s}) + (s_{r} + s_{s}) = t + s = (t + s)_{r} + (t + s)_{s} = t_{r} + s_{r} + (t + s)_{s}$ hence
$(t + s)_{s} = t_{s} + s_{s}$. For (iii) recall from Theorem \ref{regg} that the usual sum $t_r + s_r$ of two regular bilinear forms $t_r, s_r$ is regular. Hence, we have $t_r + s_r = (t_{r} + s_{r})_{r} \leq (t + s)_{r}$. (iv) follows
from the definition of the regular and singular parts.
\end{proof}

\begin{theorem}\label{rest}

Let ${\mathcal   H}$ be an infinite-dimensional complex Hilbert space.  Let us define a partial operation $\op$ on $\Vp$ by
\begin{description}
\item for  $t,s \in \Vp$, $t \op s$ is defined if and only if $t \oplus s$ is defined and $(t + s)_{r} = t_{r} + s_{r}$ and then
$t \op s = t \oplus s$.
\end{description}
Then $(\Vp;\op,o)$ is a generalized effect algebra.
\end{theorem}

\begin{proof}
First, we show the commutativity. Let us have $t,s \in \Vp$ such that $t \op s$ is defined. Then $(s+t)_{r} = (t + s)_{r} = t_{r} + s_{r} = s_{r} + t_{r}$ and $t \oplus s = s \oplus t$ is defined hence
$t \op s = s \op t$.

Let $t,s,u \in \Vp$, $(t \op s) \op u$ be defined. We have $(t \oplus s) \oplus u = t \oplus (s \oplus u)$ is defined. Then from
$(t+s)+u = t+(s+u)$ we get by Lemma \ref{restlem} (iv), $(t+(s+u))_{r} = ((t+s)+u)_{r} = t_{r} + (s_{r} + u_{r})$. Using Lemma \ref{restlem} (iii), the inequalities
$t_{r} + (s_{r} + u_{r}) \leq t_{r} + (s+u)_{r} \leq (t+(s+u))_{r}$ holds. Hence, by the previous facts,
$t_{r} + (s_{r} + u_{r}) = t_{r} + (s+u)_{r} = (t+(s+u))_{r}$. Therefore,
$t \op (s \op u)$ is defined and because $\op$ is a restriction of $\oplus,$ we have
$t \op (s \op u) = (t \op s) \op u$.

Since $t  = t + o$ and $t_r =t_{r} + o_{r} =(t + o)_{r}$ for all $t \in \Vp$, the axiom (Giii) is also satisfied. Axioms (Giv) and (Gv) are
trivial consequences of the fact, that the partial operation $\op$ is the restriction of $\oplus$.
\end{proof}

\begin{theorem}
Let ${\mathcal H}$ be an infinite-dimensional complex Hilbert space. Then $\R$ and $\S$
form sub-generalized effect algebras of the generalized effect algebra $(\Vp; \op,o).$
\end{theorem}
\begin{proof}
We have $\R \cap \S = \{ o \}$. From Theorem \ref{regg} the usual sum $t_r + s_r$ of two regular bilinear  forms $t_r, s_r$ is regular. Let
$t,u \in \R$, $s \in \Vp \setminus \B$ such that $t \op s = u$. Then $(t + s)_{r} = t_{r} + s_{r}$, hence
$t+s = u = u_{r} = (t + s)_{r} = t_{r} + s_{r} = t + s_{r}$ that is $t + s = t + s_{r}$. Since $t \oplus s$
is defined and $s$ is unbounded, we have $D(t + s) = D(s)$ hence $D(s) = D(s_{r})$ and $s = s_{r}$.

Now let us assume that, for $t,s \in \S$, $t \op s$ is defined. By Lemma \ref{restlem} (ii) $(t + s)_{s} = t_{s} + s_{s} = t + s$. The proof of closedness of the subtraction is analogous to the regular case.
\end{proof}

\begin{remark}
It is easy to see that a partial operation $\oplus$ coincides with $\op$ on $\R$ ($\S$ respectively). That is, $\oplus_{\R} = \op_{\R}$ and $\oplus_{\S} = \op_{\S}$. Note that $\R$ is a sub-generalized
effect algebra of $(\Vp; \op,o)$, but it is not a sub-generalized effect algebra of $(\Vp; \oplus,o)$ and the corresponding restricted generalized effect algebras $(\R; \oplus_{\mid \R},o)$ and $(\R; \op_{\mid \R},o)$ are isomorphic.

\end{remark}

\section{Monotone convergence}%5

In the section, we will study some monotone Dedekind upwards (downwards) $\sigma$-complete properties of different kinds of generalized effect algebras of bilinear forms corresponding to the induced order $\le :=\le_\oplus.$ It is necessary to recall that some monotone convergence theorems for bilinear forms were studied in \cite{bs, kato}, however not with respect to generalized effect algebras and their induced partial order $\le =\le_\oplus, $ but rather with respect to the partial order $\preceq$ of bilinear forms defined by (\ref{prec}).

We say that a generalized effect algebra $E$ is (i) {\it monotone Dedekind upwards $\sigma$-complete} if, for any sequence $x_1 \le x_2\le \cdots,$ which is  dominated  by some element $x_0$, i.e. each $x_n \le x_0,$ the element $x = \bigvee_n x_n$  is defined in $E$ (we write $\{x_n\}\nearrow x$), (ii) {\it monotone Dedekind downwards $\sigma$-complete} if, for any sequence $x_1 \ge x_2\ge \cdots,$  the element $x = \bigwedge_n x_n$  is defined in $E$ (we write $\{x_n\}\searrow x$). If $E$ is an effect algebra,  both later  notions are equivalent. In addition, we say that a generalized effect algebra $E$ is {\it upwards directed} if, given $a_1,a_2 \in E,$ there is $a \in E$ such that $a_1,a_2 \le a.$

\begin{lemma}\label{le:5.1}
Let  $(E,\oplus,0)$ be a generalized effect algebra. If $E$ is Dedekind upwards monotone $\sigma$-complete, then $(E,\oplus,0)$ is  Dedekind downwards monotone $\sigma$-complete. If, in addition, $E$ is upwards directed, then $E$ is Dedekind upwards monotone $\sigma$-complete if and only if $E$ is downwards monotone $\sigma$-complete.
\end{lemma}

\begin{proof}
(1) Assume $E$ is upwards monotone $\sigma$-complete.
Let us have a sequence $\{ a_n \}_{n \in \mathbb{N}} \in E$ such that $a_1 \geq a_2 \geq \cdots$. Then the non-decreasing sequence $\{a_1 \ominus a_n\}_{n \in \mathbb{N}}$  is dominated by $a_1$. Hence, there exists $a' = \bigvee_{n\in\mathbb{N}} \{a_1 \ominus a_n\}$. Since $a_1 \geq a_1 \ominus a_n$ for all $n\in\mathbb{N}$, then
$a_1 \geq a'$, i.e., $a_1 \ominus a'$ is defined. We have $a_n = a_1 \ominus (a_1 \ominus a_n) \geq a_1 \ominus a',$ hence  $a_1 \ominus a'$ is a lower bound of the sequence $\{ a_n \}_{n \in \mathbb{N}}$.
Let us assume that there exists $b \in E$ such that $b\leq a_n$ for all $n \in \mathbb{N}$. Then $a_1 \ominus b \geq a_1 \ominus a_n$ for all $n \in \mathbb{N},$
therefore $a_1 \ominus b \geq \bigvee_{n\in\mathbb{N}} \{a_1 \ominus a_n\} = a'.$ Hence $a_1 \geq a' \oplus b$ and $a_1 \ominus a'\geq b$. That is,  $a_1 \ominus a' = \bigwedge_{n \in \mathbb{N}} \{a_n\}$.

(2) Now let $E$ be upwards directed and downwards monotone $\sigma$-complete.
Assume  $\{a_n\}_{n\in\mathbb{N}}$ is a non-decreasing sequence of elements of $E.$ Let $b_1$ be any upper bound of $\{a_n\}_{n \in \mathbb{N}}.$  Then there exists $\bigwedge _{n \in \mathbb{N}} (b_1\ominus a_n)=b_1'.$ Hence, $b_1\ominus a_n \ge b'_1,$ so that $a_n = b_1 \ominus (b_1\ominus a_n) \le b_1\ominus b_1'. $ First we assert that $b_1 \ominus b'_1$ is the least upper bound of $\{a_n\}_{n \in \mathbb{N}}$ in the interval $[0,b_1]:=\{x \in E \mid 0\le x \le b_1\}.$ Indeed, if $c\in E$ is such that $a_n \le c \le b_1,$ then $b_1\ominus c \le b_1\ominus a_n$ which yields $b_1 \ominus c \le b_1'$ and $b_1\ominus b_1'\le c.$

Now if $b_2$ is another bound of $\{a_n\}_{n \in \mathbb{N}},$ we repeat the above steps and we have that $b_2 \ominus b_2',$ where  $\bigwedge_{n \in \mathbb{N}} (b_2 \ominus a_n)=b_2',$ is the least upper bound of $\{a_n\}_{n \in \mathbb{N}}$ in the interval $[0,b_2].$  Since $E$ is upwards directed, there is $b_3\in E$ such that $b_3\ge b_1,b_2.$ We repeat again the procedure also for $b_3,$ so that $b_3 \ominus b_3',$ where $b_3'=\bigwedge_{n \in \mathbb{N}}(b_3\ominus a_n),$ is the least upper bound of the sequence $\{a_n\}_{n \in \mathbb{N}}$ in the interval $[0,b_3].$ But $b_1\ominus b_1'$ and $b_2\ominus b_2'$ are also upper bounds for $\{a_n\}_{n \in \mathbb{N}}$ and both are from the interval $[0,b_3],$ we conclude $b_3\ominus b_3'\le b_1\ominus b_1', b_2\ominus b_2'.$ This entails that $b_3\ominus b_3'$ is an upper bound for $\{a_n\}_{n \in \mathbb{N}}$ in the intervals $[0,b_1]$ and $[0,b_2]$ so that, $b_1\ominus b_1'\le b_3\ominus b_3'$ and $b_2\ominus b_2'\le b_3 \ominus b_3'.$  Hence, $b_1\ominus b_1'=b_3\ominus b_3 = b_2\ominus b_2',$ which proves that $\bigvee_{n \in \mathbb{N}} a_n $ exists in $E$ and $b_1 \ominus b_1'= \bigvee_{n \in \mathbb{N}} a_n = b_2\ominus b_2'.$
\end{proof}

\begin{theorem}\label{mdd}
Let ${\mathcal H}$ be an infinite-dimensional complex Hilbert space. Then the generalized effect algebra $(\Vp; \oplus, o)$ is
monotone Dedekind downwards  $\sigma$-complete.
\end{theorem}

\begin{proof}
From Theorem \ref{subbound}, we can see that if $s, t \in \Vp$, $s \leq t$ and
$t$ is a bounded bilinear form, then $s$ is a bounded bilinear form as well.

Let, for a sequence of positive bilinear forms $\{t_n\}_{n \in \mathbb{N}}$ from $\Vp,$ we have $t_1 \geq t_2 \geq \cdots.$ The sequence is bounded from below e.g. by the bilinear form $o.$ Whenever $t_n$ is a bounded bilinear form for some $n\in \mathbb{N}$, then $t_m$ is bounded for all $m \in \mathbb{N}$, $m \geq n$. We can
use \cite[Thm VIII 3.3]{kato} which states that there exists $t \in \mathcal{PBF(H)}$ with $D(t)=\mathcal H$ such that $t \leq t_n$ for all $n \in \mathbb{N}$ and $\lim_{n\to\infty} t_n(u,v) = t(u,v)$ for every $u,v \in \mathcal{ H}$.
With Remark \ref{remord} it can be seen that $t=\bigwedge_n t_n.$

Let us consider a case when  $t_n$ is an unbounded bilinear form  for any $n \in \mathbb{N}$. According to the definition of $\oplus,$ we have $D(t_n) = D(t_m)$ for all $n,m \in \mathbb{N}$.
We show that $t$ given by $t(u,v) = \lim_{n\to\infty} t_n(u,v)$ for all $u,v \in D(t_1)$ is the requested positive bilinear form with $D(t):=D(t_1).$

The expression $t(u,u)$ is defined for all $u \in D(t_n)$ since $\{t_n(u,u) \}_{n \in \mathbb{N}}$ is for all $u \in D(t_n)$ a non-increasing sequence bounded by $0$.
Using the polarization formula on $t_n,$ we can see that also $t(u,v) \in \mathbb{C}$ is defined for every $u, v \in D(t_n)$. The bilinearity of $t$ follows from
properties of limits of sequences of complex numbers.

The sequence $\{t_n\}_{n \in \mathbb{N}}$ is non-increasing  and $t$ is its lower bound. Since $t(u,u)$ is a limit of $\{t_n(u,u)\}_{n \in \mathbb{N}},$ we have $t = \bigwedge _nt_n.$
\end{proof}

\begin{remark}
From the proof of the previous theorem we can see that a sub-generalized effect algebra $(\B; \oplus_{\mid \B}, o)$ is monotone Dedekind downwards  $\sigma$-complete.
\end{remark}

\begin{example}\label{ex:5.4}
Let $\mathcal{ H}$ be a separable complex Hilbert space with an orthonormal basis $\mathcal{ E} = \{ e_j \}_{j \in \mathbb{N}}$  and let $T$ be an operator defined
by $Te_j =: je_j$ for all $n\in \mathbb{N}$ on $D(T) = \{\sum_{j=1}^{\infty}{\alpha_j e_j} \in \mathcal{ H} \mid \sum_{j=1}^{\infty} \left|j \alpha_j \right| < \infty  \}$. $T$ is an unbounded, positive and self-adjoint linear operator. Define operators $T_n e_j:=je_j$ for every $j \leq n$ and $T_n e_j = 0$ otherwise. Clearly $T_n$ is bounded, i.e. it can be uniquely extended onto $\mathcal{ H}$ and $T_n \leq T$ for all $n \in \mathbb{N}$. Moreover, $\lim_{n\to\infty} T_n x = Tx$ for every $x \in D(T)$.

Let us take the restriction $T_{\mid span (\mathcal{ E})}$. We have also $T_n \leq T_{\mid span (\mathcal{ E})}$, but $T$ and $T_{\mid span (\mathcal{ E})}$ are not comparable.
\end{example}

Let us consider bilinear forms $t, t_{\mid span (\mathcal{ E})}, t_n$ generated by $T, T_{\mid span (\mathcal{ E})}, T_n$ for every $n \in \mathbb{N}$.
We have $t_n \leq t$ and $t_n \leq t_{\mid span (\mathcal{ E})}$ for all $n \in \mathbb{N}$, but $t$ and $t_{\mid span (\mathcal{ E})}$ are incomparable. Since both are regular and
generated by a linear operator, it follows that
$(\Vp; \oplus, o)$, $(\R; \oplus_{\mid \R}, o)$ and $(\G; \oplus_{\mid \G}, o)$ are not Dedekind upwards $\sigma$-complete. Note that $t\preceq t_{\mid span (\mathcal{ E})}$,
which makes a difference between a case considered in \cite{kato} and \cite{bs}.

On the other hand, according to Remark \ref{remord}, for any $t,s \in \Vp,$ if $t \leq s,$ then $t \preceq s$.

Recall an example from {\rm\cite{kato}}.

\begin{example}{\rm\cite[Ex VIII.3.10]{kato}}\label{exkato}
Let us have the Hilbert space $\mathcal{ H} = \mathcal{
 L}^2(0,1)$ and let us define, for every integer $n \in \mathbb N,$ a bilinear form $t_n$ by
\begin{equation}
t_n(u,u) = \frac{1}{n} \int_0^1 \left|u'(x)\right|^2 {\rm d}x + \left|u(0)\right|^2 + \left|u(1)\right|^2.
\end{equation}
Then for every $n \in \mathbb{N}$, $t_n$ is positive and closed on $D(t_n)$, which is given by $u \in D(t_n)$ whenever $u' \in \mathcal{ L}^2(0,1)$. Moreover, $t_n \geq_{\C} t_{n+1}$ since $t_n - t_{n+1}$ is a positive closed bilinear form. But
\begin{equation}
t_0:= \lim_{n \to \infty} t_n(u,u) = \left|u(0)\right|^2 + \left|u(1)\right|^2,
\end{equation}
which is known to be a singular bilinear form.

On the other hand, consider the sequence $\{ t_1 - t_n \}_{n \in \mathbb{N}}$. It is an upwards monotone  sequence of closed positive bilinear forms dominated by $t_1$, i.e. $t_1 \geq  (t_1 - t_n)$ for every ${n \in \mathbb{N}}$.
Let us set

\begin{equation}
t'(u,u) := \lim_{n \to \infty} t_1(u,u) - t_n(u,u) = \int_0^1 \left|u'(x)\right|^2 {\rm d}x.
\end{equation}

Then $t'$ is a closed bilinear form and $t' \geq (t_1 - t_n)$ for all ${n \in \mathbb{N}}$ since $t' - (t_1 - t_n)$ is defined (and closed). Note that
$t' \preceq t_1$ and also $t' \leq t_1$, but $t' \nleq_{\R} t_1$ and $t' \nleq_{\C} t_1$ since $t_1 - t'= t_0$ is a singular positive bilinear  form.

\end{example}

\begin{theorem}\label{clos}
Let ${\mathcal H}$ be an infinite-dimensional complex Hilbert space. The generalized effect algebras $(\R; \oplus_{\mid \R}, o)$ %, $(\G; \oplus_{\mid \G}, o)$
and $(\C;$ $ \oplus_{\mid \C}, o)$ are not
monotone Dedekind downwards $\sigma$-complete.
\end{theorem}

\begin{proof}
Let us consider a sequence $\{s_n\}_{n\in \mathbb{N}} \in \C$ given by $s_n: = t_1 + t_n-t_0$, where  $t_n$'s are the bilinear forms from the previous example.
Namely:
\begin{equation}
s_n(u,u) = (1+ \frac{1}{n}) \int_0^1 \left|u'(x)\right|^2 {\rm d}x + \left|u(0)\right|^2 + \left|u(1)\right|^2.
\end{equation}
Then $s_n \geq_{\mid \C} s_{n+1}$ %($s_n \geq_{\mid \G} s_{n+1}$ and
($s_n \geq_{\mid \R} s_{n+1},$ respectively) for all $n \in \mathbb{N}$. Clearly $t_1 \leq_{\mid \C} s_n$ and $t' \leq_{\mid \C} s_n$,
but $t_1$ and $t'$ are mutually incomparable in the ordering given by $\leq_{\mid \C}$ %, $\leq_{\mid \G}$
and $\leq_{\mid \R}$. Since $t'\leq t_1$ and
$t_{1} = \bigwedge_{n\in \mathbb{N}} s_n$ in $(\Vp;\oplus,o),$ they also have no join $r$  which would satisfy
$r \leq_{\mid \C} s_n$ for all $n \in \mathbb{N}$. That is, the infimum of $\{s_n\}_{n \in \mathbb{N}}$ does not exist.
\end{proof}

\begin{theorem}\label{up}
Let ${\mathcal   H}$ be an infinite-dimensional complex Hilbert space and $D \subseteq \mathcal{ H}$ its dense linear subspace. Let us define the set $\VD \subseteq \Vp$ by

$$ \VD = \{t \in \Vp \mid t \text{ is bounded, or }  D(t) = D \}. $$
Then $\VD$ is a sub-generalized effect algebra of $(\Vp; \oplus, o),$  and  $(\VD$; $\oplus_{\mid \VD}, o),$ with a total operation $\oplus_{\mid \VD},$ is monotone Dedekind upwards and downwards $\sigma$-complete.
\end{theorem}

\begin{proof}

Since $o$ is bounded, $o \in \VD$.  If $t,s \in \VD,$ then always $t+s$ is defined in $\VD,$ so that $t\oplus_{\mid \VD} s = t+s,$ and $\oplus_{\mid \VD}$ is a total operation. In addition, we note that $\oplus_{\mid \VD}$ is the restriction of $\oplus$ from $\Vp.$  This also implies that $\VD$ is upwards directed.

Since $s\le_{\mid \VD} t$ if and only if $s \le t,$ we note that $\VD$ is a sub-generalized effect algebra of $(\Vp; \oplus, o).$ By Theorem \ref{mdd}, we have that also $\VD$ is monotone Dedekind downwards $\sigma$-complete. In view of Lemma \ref{le:5.1}, $\VD$ is  monotone Dedekind upwards $\sigma$-complete, too.
\end{proof}

\begin{theorem}\label{th:5.8}
Let ${\mathcal H}$ be an infinite-dimensional complex Hilbert space. Then the generalized effect algebra $(\Vp; \op, o)$ is not
monotone Dedekind downwards $\sigma$-complete.
\end{theorem}

\begin{proof}
Let us consider the sequence $\{ s_n \}_{n \in \mathbb{N}}$ from the proof
of Theorem \ref{clos}. Since $\R$ is a sub-generalized effect algebra of $(\Vp; \op, o)$,
we have $s_n \geq_{\op} s_{n+1}$, $s_n \geq_{\op} t_1$ and $s_n \geq_{\op}
t'$ for all $n \in \mathbb{N}$. Inequality $t' \leq t_1$ gives that
$t_1 - t' \in \Vp$. Then $t_1 = (t_1)_r = ((t_1 - t') + t')_r \not= (t_1 -
t')_r + (t')_r = o + (t')_r = t'$. Hence, $(t_1-t') \op t'$ is not defined
and we have $t' \nleq_{\op} t_1$. Considering $t' \leq t_n$, we can see
that $t_1$ and $t'$ are incomparable in the ordering given by $\leq_{\op}$.
Since $t_{1} = \bigwedge_{n\in \mathbb{N}} s_n$ in $(\Vp;\oplus,o)$, they
also have no join $r$ which would satisfy
$r \leq_{\op} s_n$ for all $n \in \mathbb{N}$. That is, the infimum of $\{s_n\}_{n \in \mathbb{N}}$ does not exist in $(\Vp; \op, o)$.
\end{proof}

Lemma \ref{le:5.1} implies that if a generalized effect algebra $E$ is not monotone Dedekind downwards $\sigma$-complete, then it is not monotone Dedekind upwards $\sigma$-complete. Therefore, from the previous theorems we conclude the following corollary.

\begin{corollary}
Let ${\mathcal H}$ be an infinite-dimensional complex Hilbert space.

{\rm (1)} The generalized effect algebra $(\VD; \oplus, o)$ is monotone Dedekind upwards and downwards $\sigma$-complete. {\rm Lemma \ref{le:5.1}, Theorem \ref{up}.}

{\rm (2)} The generalized effect algebra $(\Vp; \oplus, o)$ is not Dedekind upwards $\sigma$-complete, but it is monotone Dedekind downwards $\sigma$-complete. {\rm Example \ref{ex:5.4}, Theorem \ref{mdd}.}

{\rm (3)} The generalized effect algebras $(\R; \oplus_{\mid \R}, o),$
%, $(\G; \oplus_{\mid \G}, o)$,
$(\C;$ $ \oplus_{\mid \C}, o)$ and $(\Vp; \op, o)$ are neither monotone Dedekind upwards $\sigma$-complete,
nor monotone Dedekind downwards $\sigma$-complete. {\rm Lemma \ref{le:5.1}, Theorem \ref{clos}.}
\end{corollary}

In what follows, we show that if instead of the order $\le :=\le_\oplus$ induced by the addition $\oplus$ in a generalized effect algebra of positive bilinear forms we use the order $\preceq$ defined by (\ref{prec}), we can have a monotone upwards Dedekind $\sigma$-complete poset.

\begin{theorem}\label{th:close}
Let $\mathcal H$ be an infinite-dimensional complex Hilbert space. Then $(\C; \oplus_{\C},o)$ is a monotone Dedekind upwards $\sigma$-complete poset under the partial order $\preceq,$  but which is not a monotone Dedekind upwards $\sigma$-complete generalized effect algebra.
\end{theorem}

\begin{proof}
Assume that we have a sequence $\{t_n\}_{n \in \mathbb{N}}$ of positive closed bilinear forms which is dominated by $t_0\in \C,$ that is,  $t_1 \preceq t_2\preceq \cdots \preceq t_0.$ Therefore, we have $D(t_0)\subseteq D(t_{n+1}) \subseteq D(t_n).$ If we denote $D(t)=\{x \in \bigcap_n D(t_n) \mid \sup_n t_n(x,x)< \infty\},$ then by \cite[Thm 3.1]{bs}, the function $t(x,x)=\lim_nt_n(x,x)$ defines a closed bilinear form $t$ such that $D(t_0)\subseteq D(t).$ It yields $t_n \preceq t$ for any integer $n \in \mathbb N.$

Now let $t'\in \C$ be such that $t_n \preceq t'$ for any $n \in \mathbb N.$ Then as in the first part, we have $D(t') \subseteq D(t)$ and $t(x,x)\le t'(x,x)$ for any $x \in D(t'),$ which proves that $t \preceq t',$ and $t$ is the least upper bound of $\{t_n\}_{n \in \mathbb{N}}$ with respect to the order $\preceq.$
\end{proof}

\end{document}